\documentclass[fleqn,10pt]{wlscirep}
\usepackage{graphicx}
\usepackage{amsmath}
\usepackage{amssymb}
\usepackage{mathrsfs}
\usepackage{amsthm}
\usepackage{bm}
\usepackage{url}
\usepackage{csquotes}
\MakeOuterQuote{"}
\newtheorem{theorem}{Theorem}

\newtheorem{corollary}{Corollary}
\newcommand{\mrm}[1]{\mathrm{#1}}
\def\vec#1{\bm{#1}} 
\newcommand{\tr}{\operatorname{tr}}

\newcommand{\diag}{\operatorname{diag}}

\newcommand{\rmT}{\mathrm{T}}

\title{Steering Bell-diagonal states}

\author[1]{Quan Quan}
\author[2,3*]{Huangjun Zhu}
\author[$1^\ddag$]{Si-Yuan Liu}
\author[4,5]{Shao-Ming Fei}
\author[6,7,1]{Heng Fan}
\author[1,8]{Wen-Li Yang}
\affil[1]{Institute of Modern Physics, Northwest University, Xi'an 710069, China}
\affil[2]{Institute for Theoretical Physics, University of Cologne,  Cologne 50937, Germany}
\affil[3]{Perimeter Institute for Theoretical Physics,  Waterloo,  Ontario N2L 2Y5, Canada}
\affil[4]{School of Mathematical Sciences, Capital Normal University, Beijing 100048, China}
\affil[5]{Max-Planck-Institute for Mathematics in the Sciences,
Leipzig 04103, Germany}
\affil[6]{Institute of Physics, Chinese Academy of Sciences, Beijing 100190, China}
\affil[7]{Collaborative Innovation Center of Quantum Matter, Beijing 100190, China}
\affil[8]{Center for Mathematics and Information Interdisciplinary Sciences, Beijing,
100048, China}
\affil[*]{Corresponding author: hzhu1@uni-koeln.de, hzhu@pitp.ca }
\affil[$^\ddag$]{lsy5227@163.com}

\begin{abstract}
We investigate the steerability of  two-qubit Bell-diagonal states under  projective measurements by the steering party. In the  simplest nontrivial scenario of two projective measurements, we solve this problem completely by virtue of the connection between the steering problem and the joint-measurement problem.
A  necessary and sufficient criterion is derived together with
 a simple geometrical interpretation. Our study shows that a Bell-diagonal state is steerable by two projective measurements iff it violates the
Clauser-Horne-Shimony-Holt (CHSH) inequality, in sharp contrast with the  strict hierarchy expected between steering and Bell nonlocality. We also introduce a steering measure  and clarify its connections with concurrence and the volume of the steering ellipsoid.
 In particular, we determine the maximal concurrence and ellipsoid volume of  Bell-diagonal states that are not steerable by two projective measurements. Finally, we explore the steerability of Bell-diagonal states under three projective measurements. A simple  sufficient criterion is derived, which
can detect the steerability of many states that are not steerable by two projective measurements.  Our study offers valuable insight on steering of Bell-diagonal states as well as the connections between entanglement, steering, and Bell nonlocality.
\end{abstract}

\begin{document}

\flushbottom
\maketitle

\thispagestyle{empty}

\section*{Introduction}

Einstein-Podolsky-Rosen (EPR) steering~\cite{EinsPR35}, as noticed by Schr\"{o}dinger~\cite{Schr35}, is an intermediate type of nonlocal
correlations that sits between entanglement and Bell nonlocality. In the framework of modern
quantum information theory, this "spooky action" can be described as a task of entanglement verification
with an untrusted party, as explained by Wiseman \emph{et
al.} ~\cite{WiseJD07, JoneWD07}. It hinges on the question of whether Alice can convince
Bob that they share an entangled state, despite the
fact that Bob does not trust Alice. In order to achieve this task, Alice  needs to change Bob's
state remotely in a way that would be impossible if they shared classical
correlations only.  Contrary to
entanglement and Bell nonlocality,  steering features a fundamental asymmetry because   the two observers play different roles in the steering test \cite{WiseJD07, JoneWD07, BowlVQB14}. Recently, growing attention  has been directed to  steering  because of its potential applications in quantum information processing, such as quantum key
distribution (QKD)~\cite{BranCWSW12}, secure teleportation~\cite{Reid13}, and entanglement assisted subchannel
discrimination~\cite{PianW15}.

Two basic questions concerning steering  are its detection and quantification. One approach for detecting steering is to  prove the impossibility of constructing any non-steering model ~\cite{WiseJD07, JoneWD07}. A practical alternative is to demonstrate the violations of various steering inequalities  \cite{Reid89, CavaJWR09, WalbSGT11,Puse13, KogiSCA15, ZhuHC15}.    The first steering inequality was derived by Reid in 1989~\cite{Reid89}, which is applicable to continuous variable systems, as considered in EPR's original argument. General theory of experimental steering criteria were  developed in Ref.~\cite{CavaJWR09}, followed by many other works \cite{WalbSGT11,  Puse13, KogiSCA15, ZhuHC15}. In line with theoretical development, a loophole-free steering experiment was  reported in Ref.~\cite{Witt12}, and one-way steering was demonstrated in Ref.~\cite{Hand12}.
In addition, steering detection based on all-versus-nothing argument was  proposed in  Refs.~\cite{ChenYWS13, WuCYS13}, along with an experimental demonstration \cite{SunXYW14}. Meanwhile, steering quantification has received increasing attention in the past few years \cite{SkrzNC14, PianW15, KogiLRA15}, which leads to   several useful steering measures, such as steerable weight \cite{SkrzNC14} and  steering robustness~\cite{PianW15}.

Despite these fruitful achievements, steering detection and quantification have remained challenging tasks, and many basic questions are poorly understood. For example, no conclusive criterion is known for determining the steerability of  generic two-qubit states except for Werner states \cite{WiseJD07, JoneWD07}.
Even for  Bell-diagonal
states,  only a few partial results are known concerning their steerability,
including
several necessary criteria and several sufficient
criteria \cite{ChenSYWO12,JevtPJR14,JevtHAZW15}; further progresses are thus highly desirable.
In addition, many results in the literature rely heavily on numerical calculation and lack intuitive pictures. Analytical results are quite rare since difficult optimization problems are often involved in solving steering problems.

In this work, we investigate the steerability
of  two-qubit Bell-diagonal states under projective measurements by the steering party. These states are appealing to both theoretical and experimental studies since they have a relatively simple structure and are particularly suitable for illustrating ideas and cultivating intuition. In addition, generic two-qubit states can be turned into  Bell-diagonal
states by invertible stochastic local operation and classical communication
(SLOCC) \cite{VersDM01}, so any progress on  Bell-diagonal states may potentially help understand two-qubit states in general.

We first consider the steerability
of  Bell-diagonal states under the simplest nontrivial measurement
setting on the steering party, that is, two projective measurements. We solve this problem completely by virtue of the connection between the steering problem and the joint-measurement problem \cite{QuinVB14, UolaMG14, UolaBGP15, ZhuHC15}.  In particular, we derive a necessary and sufficient steering criterion analytically and provide a simple geometrical interpretation. Such analytical results are valuable but quite rare in the literature on steering. Our study  leads to a measure of steering, which turns out to equal the maximal violation of the Clauser-Horne-Shimony-Holt (CHSH)  inequality \cite{ClauHSH69, HoroHH95}. As an implication, a Bell-diagonal state is steerable by two projective measurements iff it violates the CHSH inequality. This conclusion presents a sharp contrast with the observation that steering is  necessary but usually not sufficient for Bell nonlocality \cite{WiseJD07, JoneWD07, QuinVCA15}. On the other hand, in the special case of rank-2 Bell-diagonal states,  entanglement is sufficient to guarantee steering and Bell nonlocality, in line with the spirit of Gisin's theorem \cite{Gisi91, ChenSXW15}.
The relations between our steering measure and concurrence as well as the volume of the steering ellipsoid
are  then clarified. Quite surprisingly, the steering measure and the volume of the steering ellipsoid seem to display opposite behaviors for states with given concurrence.

Finally, we explore the steerability of Bell-diagonal states under three projective measurements. Although such problems are generally very difficult to address, we
derive a nontrivial  sufficient criterion, which also has a simple geometrical interpretation. This criterion can detect the steerability of many states that are not steerable by two projective measurements. The relation between entanglement and steering in this scenario is also clarified.

\section*{Setting up the stage}

Consider  two remote parties, Alice and Bob, who share a bipartite quantum state $\rho$ with reduced states $\rho_\mathrm{A}$ and $\rho_\mathrm{B}$ for the two parties, respectively. Alice can perform a collection of local measurements as characterized by a collection of  positive-operator-valued measures (POVMs) $\{A_{a|x}\}_{a,x}$, where $x$ labels the POVM and $a$ labels the outcome in each POVM.  Recall that a POVM $\{A_{a|x}\}_a$ is composed of a set of positive operators that   sum up to the identity, that is, $\sum_a A_{a|x}=I$. The whole collection of POVMs $\{A_{a|x}\}_{a,x}$ is called  a  \emph{measurement assemblage}. If Alice  performs the measurement $x$ and obtains the outcome $a$, then  Bob's subnormalized reduced state is given by $\rho_{a|x}=\tr_\mathrm{A}[(A_{a|x}\otimes
I)\rho]$. Note that $\sum_a\rho_{a|x}=\rho_\mathrm{B}$ is  independent  of the measurement chosen by Alice, as required by the no signaling principle. The set of subnormalized states $\{\rho_{a|x}\}_a$ for a given measurement $x$ is an \emph{ensemble}  for $\rho_\mathrm{B}$, and the whole collection of ensembles $\{\rho_{a|x}\}_{a,x}$ is a \emph{state assemblage} \cite{Puse13}.

The  state assemblage $\{\rho_{a|x}\}_{a,x}$ is  \emph{unsteerable} if there exists a local hidden state (LHS) model~\cite{WiseJD07, JoneWD07, QuinVB14, UolaMG14, UolaBGP15, ZhuHC15}:
\begin{equation}
\rho_{a|x}=\sum_\lambda p_\rho(a|x,\lambda)\sigma_\lambda,
\end{equation}
 where $p_\rho(a|x,\lambda)\geq0$, $\sum_a p_\rho(a|x,\lambda)=1$, and $\sigma_\lambda$ are a collection of subnormalized states that sum up to $\rho_\mathrm{B}$ and thus form an ensemble for $\rho_\mathrm{B}$.  This model means that Bob can interpret his conditional states $\rho_{a|x}$  as coming from the preexisting states $\sigma_{\lambda},$ where only the probabilities are changed due  to the knowledge of Alice's measurements and outcomes.

The steering problem is closely related to the joint-measurement problem. A measurement assemblage $\{A_{a|x}\}_{a,x}$ is \emph{compatible} or \emph{jointly measurable} \cite{HeinW10,
QuinVB14, UolaMG14, UolaBGP15, Zhu15IC} if there exist a POVM $\{G_\lambda\}$ and probabilities $p_\mathrm{A}(a|x,\lambda)$ with $\sum_a p_\mathrm{A}(a|x,\lambda)=1$ such that
\begin{equation}
A_{a|x}=\sum_\lambda p_\mathrm{A}(a|x,\lambda)G_\lambda.
\end{equation}
Physically, this means that all the measurements in the assemblage can be measured jointly by performing  the measurement $\{G_\lambda\}_\lambda$ and  post processing  the measurement data. According to the above discussion, determining the compatibility of a measurement assemblage
is mathematically equivalent to
determining the unsteerability of a state assemblage.
Therefore, many compatibility problems can be translated into steering  problems, and vice versa \cite{QuinVB14, UolaMG14, UolaBGP15, ZhuHC15}.
This observation
will play an important role in the present study.

When $\rho_\mathrm{B}$ is of full rank, the state assemblage $\{\rho_{a|x}\}_{a,x}$ for Bob can be turned into a measurement assemblage as follows \cite{UolaBGP15, ZhuHC15},
\begin{equation}
B_{a|x}=\rho_\mathrm{B}^{-1/2}\rho_{a|x}\rho_\mathrm{B}^{-1/2}.
\end{equation}
Note that the set of operators $\{B_{a|x}\}_a$ for a given $x$ forms a POVM, which is referred to as Bob's \emph{steering-equivalent  observable} (or POVM)~\cite{UolaBGP15}.
The measurement assemblage $\{B_{a|x}\}_{a,x}$ is compatible iff the state assemblage $\{\rho_{a|x}\}_{a,x}$ is unsteerable. For example,  if $\rho_{a|x}=\sum_\lambda p(a|x,\lambda)\sigma_\lambda$, then $B_{a|x}=\sum_\lambda p(a|x,\lambda)G_\lambda$ with $G_\lambda=\rho_\mathrm{B}^{-1/2}\sigma_\lambda\rho_\mathrm{B}^{-1/2}$; the converse  follows from the same reasoning. This observation suggests a fruitful approach for understanding steering  via  steering-equivalent observables.

\section*{Results}

\subsection*{Steer Bell-diagonal states by projective measurements}

Any   two-qubit state can be  written in the following form
 \begin{equation}
 \label{eq:qub}
\rho=\frac{1}{4}(I\otimes I+\vec{a}\cdot\vec{\sigma}\otimes I+I\otimes\vec{b}\cdot\vec{\sigma}+\sum_{i,j=1}^3t_{ij}\sigma_i\otimes\sigma_j),
\end{equation}
where  $\sigma_j$ for $j=1,2,3$  are three Pauli
matrices, $\vec{\sigma}$ is the vector composed of these Pauli matrices, $\vec{a}$ and $\vec{b}$ are the Bloch vectors associated with the  reduced states of Alice and Bob, respectively, and $T=(t_{ij})$  is the correlation  matrix. The   two-qubit  state is a Bell-diagonal state iff $\vec{a}=\vec{b}=\vec{0}$ \cite{HoroH96}, in which case we have
\begin{equation}\label{eq:BellDiagonal}
\rho=\frac{1}{4}(I\otimes I+\sum_{i,j=1}^3t_{ij}\sigma_i\otimes
\sigma_j),
\end{equation}
with two completely mixed marginals, that is,
$\rho_\mathrm{A}=\rho_\mathrm{B}=I/2$. Bell-diagonal states are of special interest because they have a simple structure  and are thus a good starting point for understanding states with more complex structure. In addition,  all two-qubit states except for a set of measure zero can be turned into  Bell-diagonal
states by invertible SLOCC \cite{VersDM01}.

With a suitable local unitary transformation, the correlation matrix $T$     in \eqref{eq:BellDiagonal}  can be turned into diagonal form, so that
\begin{equation}\label{eq:BellDiagDiag}
\rho=\frac{1}{4}(I\otimes I+\sum_{j=1}^3t_j\sigma_j\otimes
\sigma_j).
\end{equation}
As an implication of this observation, a Bell-diagonal state is steerable by one party iff it is steerable by the other party, so there is no one-way steering \cite{BowlVQB14}
for Bell-diagonal states. It does not matter which party serves as the steering
party in the present study.

In the case of a qubit, any projective measurement $\{A_{\pm|x}\}$ with two outcomes $\pm$ is uniquely determined by a unit vector $\vec{e}_x$ on the Bloch sphere as  $A_{\pm|x}=(I\pm \vec{e}_x\cdot \vec{\sigma})/2$. If Alice and Bob share the Bell-diagonal state \eqref{eq:BellDiagonal} and Alice performs the projective measurement determined by $\vec{e}_x$, then the two outcomes will occur with the same probability of $1/2$, and the subnormalized reduced states of Bob are given by
$\rho_{\pm|x}=[I\pm(T^\rmT \vec{e}_x)\cdot \vec{\sigma}]/4$.
Accordingly, Bob's steering-equivalent observable takes on the form
\begin{equation}\label{eq:SEqubit}
B_{\pm|x}= \frac{1}{2}(I\pm\vec{r}_x\cdot \vec{\sigma}),\quad \vec{r}_x=T^\rmT \vec{e}_x.
\end{equation}
Note that this observable is uniquely characterized by the subnormalized vector $\vec{r}_x$, which determines an unbiased noisy (or unsharp) von Neumann observable. Here "unbiased" means that $\tr(B_{+|x})=\tr(B_{-|x})=1$. In this way, the correlation matrix $T$ induces a map from projective measurements of Alice to noisy projective measurements of Bob.
Alice can steer Bob's system using the measurement assemblage  $\{A_{\pm|x}\}_x$ iff the set of noisy projective measurements $\{B_{\pm|x}\}_x$ is  incompatible.

\begin{figure}[tbh]
\centering
\includegraphics[scale=0.1]{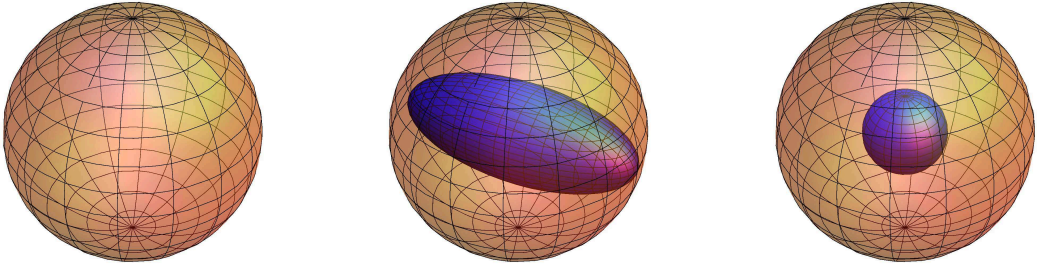}
\caption{The steering ellipsoids of three Bell-diagonal states. The ellipsoid of a Bell state (left) coincides with the Bloch sphere; the ellipsoid of a rank-2 Bell-diagonal state or an edge state (middle) is rotationally symmetric with the largest semiaxis equal to the radius of the Bloch sphere; the ellipsoid of a Werner state (right) is a sphere contained in the Bloch sphere.}
\label{fig:ellip}
\end{figure}

To see the geometric meaning of the map induced by $T$, note that the end point of  $\vec{r}_x$ lies on an ellipsoid $\mathcal{E}$ centered at origin and characterized by the symmetric matrix $T^\rmT T$: the three eigenvalues of $T^\rmT T$ are the squares of the three semiaxes (some of which may vanish), and the eigenvectors determine the orientations of these semiaxes; see Fig.~\ref{fig:ellip}. This ellipsoid encodes the set of potential noisy projective measurements of Bob induced by projective measurements of Alice. Geometrically, this ellipsoid is identical to the steering ellipsoid introduced in Refs.~\cite{Vers02,ShiJSD11,JevtPJR14}, which encodes the set of states to which Alice can steer Bob's system. It is also referred to as the steering ellipsoid
here although the meaning is slightly different from that in Ref.~\cite{Vers02,ShiJSD11,JevtPJR14}.  Since its discovery, the steering ellipsoid has played an important role in understanding various features pertinent to  entanglement and steering \cite{Vers02,ShiJSD11,JevtPJR14, MilnJJWR14,JevtHAZW15}.
To appreciate its significance in the current context, note that the steerability of a Bell-diagonal state by   the measurement assemblage $\{A_{\pm|x}\}_x$ is completely determined by the set of vectors $\vec{r}_x$ on the steering ellipsoid.
Moreover, in several cases of primary interest to us, the steerability can be determined by purely geometrical means, as we shall see shortly.

\subsection*{Steering by two projective measurements}

In this section we derive a necessary and sufficient criterion on the steerability of a Bell-diagonal state under two projective measurements.
We also introduce a steering measure  and illustrate its geometrical meaning. Our study shows that a Bell-diagonal state is steerable by two projective measurements iff it violates the CHSH inequality. Furthermore, we clarify the relations between entanglement, steering, and Bell nonlocality
by deriving tight inequalities  between the following three measures: the concurrence, the steering measure, and the volume of the steering ellipsoid.

\begin{theorem}\label{thm:SBD2} A  Bell-diagonal state with correlation matrix  $T$ is steerable by  two projective measurements iff $\lambda_1+\lambda_2>1$,
 where $\lambda_1, \lambda_2$ are the two largest eigenvalues of $T T^\rmT$.
\end{theorem}
\begin{proof}
Suppose Alice  performs two projective measurements $\{A_{\pm|x}\}_{x=1}^2=\{(I\pm\vec{e}_x\cdot\vec{\sigma})/2\}_{x=1}^2$. Then  Bob's steering equivalent observables are given by  $\{B_{\pm|x}\}_{x=1}^2= \{(I\pm\vec{r}_x\cdot \vec{\sigma})/2\}_{x=1}^2$, where  $\vec{r}_x=T^\rmT\vec{e}_x$, as specified in~\eqref{eq:SEqubit}. According to Ref.~\cite{Busc86} (see also Refs.~\cite{StanRH08,BuscS10,
YuLLO10, Zhu15IC}), the two observables are compatible iff
\begin{equation}\label{eq:compatibility}
 |\vec{r}_1+\vec{r}_2|+|\vec{r}_1-\vec{r}_2|\leq2.
\end{equation}
Note that $\vec{r}_1$ and $\vec{r}_2$ are two vectors on the steering ellipsoid, and the left hand side of the inequality is half of the perimeter of a parallelogram  inscribed on the steering ellipsoid, with the plane spanned by the parallelogram passing the centre of the ellipsoid. So the Bell-diagonal state is steerable iff the maximal perimeter of such  parallelograms is larger than 4. Interestingly,  the maximum can be derived with a similar method used for deriving the maximal violation of the CHSH inequality  \cite{HoroHH95, Scar12},
\begin{align}
&\max_{\vec{e_1},\vec{e_2}}\{|\vec{r}_1+\vec{r}_2|+|\vec{r}_1-\vec{r}_2|\}=\max_{\vec{e_1},\vec{e_2}}\{|T^\rmT(\vec{e_1}+\vec{e_2})|+|T^\rmT(\vec{e_1}-\vec{e_2})|\}=\max_{\chi,\vec{c},\vec{c}^\bot}\{2\cos{\chi}|T^\rmT\vec{c}|+2\sin{\chi}|T^\rmT\vec{c}^\bot|\}\nonumber\\
&=2\max_{\vec{c},\vec{c}^\bot}\sqrt{|T^\rmT\vec{c}|^2+|T^\rmT\vec{c}^\bot|^2}=2\max_{\vec{c},\vec{c}^\bot}\sqrt{\vec{c}^\rmT T T^\rmT \vec{c}+\vec{c}^{\bot \rmT}T T^\rmT \vec{c}^\bot}=2\sqrt{\lambda_1+\lambda_2},\label{eq:Max}
\end{align}
where $2\chi$ is the angle spanned by $\vec{e}_1$ and $\vec{e}_2$; $\vec{c}$ and $\vec{c}^\bot$ are the direction vectors of $(\vec{e_1}+\vec{e_2})$ and  $(\vec{e_1}-\vec{e_2})$, respectively, which are always orthogonal. Here the maximum in the last step is attained  when  $\vec{c}$ and $\vec{c}^\bot$
span the same space as that spanned by the two eigenvectors associated with the two largest eigenvalues of $T T^\rmT$. The maximum over $\vec{e}_1$ and $\vec{e}_2$ can be attained when the two vectors are  eigenvectors corresponding to the two largest  eigenvalues of $TT^\rmT$.
The Bell-diagonal state is steerable by two  projective measurements iff $2\sqrt{\lambda_1+\lambda_2}>2$, that is, $\lambda_1+\lambda_2>1$.
\end{proof}

The choices of $\vec{c}$ and $\vec{c}^\bot$ that maximize \eqref{eq:Max} are highly not unique. Therefore, the optimal projective measurements that Alice needs to perform are also not unique. Although the optimal measurements can always be chosen to be mutually unbiased as shown in the above proof,  it is usually not necessary to do so. As an example, consider the Bell-diagonal state characterized by the correlation matrix $T=\diag(t_1,t_2, t_3)$ with $t_1\geq t_2\geq |t_3|$. One choice of $\vec{c}$ and $\vec{c}^\bot$ reads $\vec{c}=(1, 0, 0)$ and  $\vec{c}^\bot=(0,1, 0)$, which leads to the  optimal measurement directions  $\vec e_1=(t_1,t_2, 0)/\sqrt{t_1^2+t_2^2}$ and $\vec e_2=(t_1,-t_2, 0)/\sqrt{t_1^2+t_2^2}$. Note that  $\vec e_1$ and  $\vec e_2$ are  not orthogonal in general, so the corresponding projective measurements are not mutually unbiased.

The proof of Theorem~\ref{thm:SBD2} also suggests  a steering measure  of a Bell-diagonal state  under two projective
measurements,  namely,   $S:=2\sqrt{\lambda_1+\lambda_2}$. This measure has a simple geometrical meaning: $(S/2)^2$ is equal to the sum of squares of the two largest semiaxes of the steering ellipsoid. A Bell-diagonal state is steerable in this scenario iff $S>2$. The maximum $2\sqrt{2}$ of $S$ is attained when $\lambda_1=\lambda_2=1$, which corresponds to a Bell state. To obtain a normalized measure of steering, we may opt for $\max\{0,(S-2)/(2\sqrt{2}-2)\}$. According to Ref.~\cite{HoroHH95}, the maximal violation of the CHSH inequality by the Bell-diagonal state is  equal to $2\sqrt{\lambda_1+\lambda_2}$ (cf.~Ref.~\cite{MilnJJR14} for a geometrical interpretation), which coincides with the steering measure  $S$ introduced here.
This observation has an important implication for the relation between steering and Bell nonlocality.
\begin{corollary}
A Bell-diagonal state is steerable by two projective measurements iff it violates the CHSH inequality.
\end{corollary}

To clarify the geometric meaning of Theorem~\ref{thm:SBD2} and the steering measure $S$, it is convenient to choose a concrete Bell basis. Here we shall adopt the following choice \cite{LangC10}, \begin{equation}
|\beta_{\mu\nu}\rangle=\frac{1}{\sqrt{2}}(|0,\nu\rangle+(-1)^\mu|1,1\oplus
\nu\rangle),\quad \mu,\nu=0,1.
\end{equation}
Note that $|\beta_{11}\rangle$ is the singlet. Thanks to the choice of the Bell basis, the correlation matrices of the four Bell states are diagonal as given by  $\diag((-1)^\mu, -(-1)^{\mu+\nu}, (-1)^\nu)$.
Up to a local unitary transformation, any Bell-diagonal state is a mixture of the four Bell states. Without loss of generality, we can focus on Bell-diagonal states of this form, whose  correlation matrices are also
diagonal, as in \eqref{eq:BellDiagDiag}.

\begin{figure}[ht]
\centering
\includegraphics[scale=0.1]{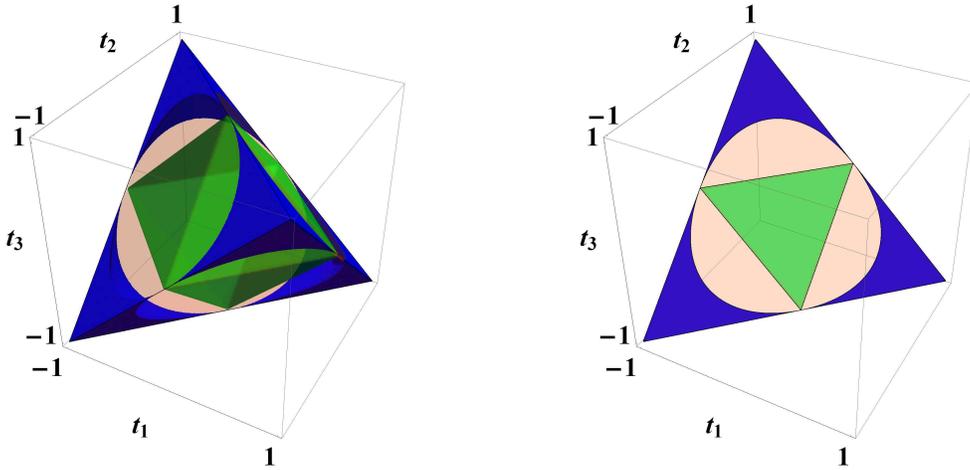}
\caption{\label{fig:SBDtwo} Geometric illustration of Bell-diagonal
states steerable by two  projective measurements.    (left) The regular tetrahedron  represents the set of  Bell-diagonal states.  The octahedron
in green represents the set of separable  states. The blue regions represent those states that  are  steerable  by two projective measurements. (right) A face of the regular tetrahedron which represents the set of rank-3 Bell-diagonal states.}
\end{figure}

Geometrically, the set of Bell-diagonal states forms a regular tetrahedron, whose vertices correspond to the four Bell states \cite{HoroH96,LangC10}.
The set of separable Bell-diagonal states forms an octahedron inside the tetrahedron   \cite{HoroH96,LangC10}.    The tetrahedron can be embedded into a cube whose sides are aligned with the three axes labelled by $t_1,
t_2, t_3$,  as shown in Fig.~\ref{fig:SBDtwo}. In this way, a Bell-diagonal state is uniquely specified by its three coordinates $(t_1, t_2, t_3)$. The
half steering measure $S/2$ of this Bell-diagonal state is equal to the maximum over $\sqrt{t_1^2+t_2^2}$, $\sqrt{t_2^2+t_3^2}$, and $\sqrt{t_3^2+t_1^2}$, which  is equal to the maximal length of  the three projections of $(t_1, t_2, t_3)$ onto the three coordinate planes. Note that $S$ is convex in $t_1, t_2, t_3$ and defines a norm in the three-dimensional vector space that accommodates Bell-diagonal states. Each level surface of this norm is determined by  three orthogonal cylinders of equal radius. In particular, the   set of unsteerable Bell-diagonal states (determined by the level surface with $S=2$) is contained in the intersection of the three solid cylinders  specified by the following three inequalities,
respectively,\begin{equation}
t_{1}^2+t_{2}^2\leq1,\quad
t_{2}^2+t_{3}^2\leq1,\quad
t_{3}^2+t_{1}^2\leq 1.
\end{equation}

\begin{figure*}[tbh]
\centering
\includegraphics[scale=0.144]{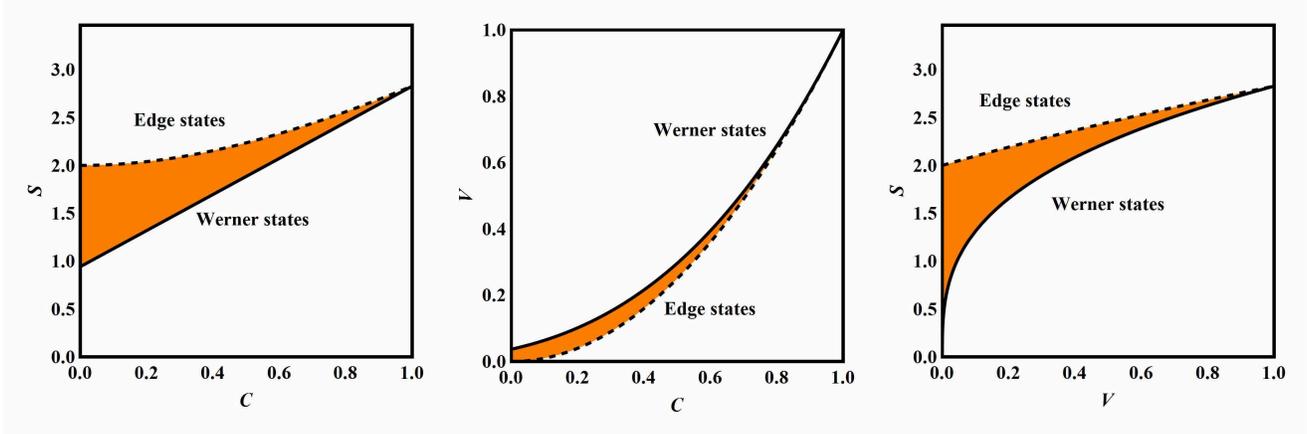}
\caption{\label{fig:SVC2}Relations between three entanglement and steering measures for Bell-diagonal states. Here $C$ is the concurrence, $S$ is the steering measure, and $V$ is the normalized volume of the steering ellipsoid.  The orange region in each plot indicates the  range of values. The dashed lines represent  edge states and the solid lines represent  Werner states.  }
\end{figure*}

In the rest of this section we clarify the relations between the following three measures: the concurrence, the steering measure $S, $ and the volume of the steering ellipsoid. Since $S$ is equal to the maximal violation of the CHSH inequality, our discussion is also of interest to studying Bell nonlocality.

Recall that a two-qubit state is entangled iff it has nonzero concurrence and that the concurrence of a Bell-diagonal state is given by  $C=\max\{0,2p_\mrm{max} -1\}$, where $p_\mrm{max}$ is the maximal eigenvalue of the state~\cite{HillW97}.   Given a Bell-diagonal state
with correlation matrix $T$, the normalized volume $V$ of the steering ellipsoid  is defined as $V:=|\det(T)|$ \cite{JevtPJR14}. If $T$ is diagonal, say
$T=\diag(t_1,t_2,t_3)$, then $V=|t_1t_2t_3|$. The constraints   $|t_j|\leq 1$ for $j=1,2,3$ imply that $0\leq V\leq 1$, where  the upper bound is saturated only for Bell states.

Calculation shows that $C,S,V$ satisfy the following inequalities (see  Methods section for more details):
\begin{align}
\frac{2\sqrt{2}}{3}(1+2C)&\leq S\leq 2\sqrt{1+C^2}, \label{eq:SC}\\
C^2&\leq V\leq \Bigl(\frac{1+2C}{3}\Bigr)^3, \label{eq:VC}\\
2\sqrt{2}\sqrt[3]{V}&\leq S\leq 2\sqrt{1+V}.\label{eq:SV}
\end{align}
Here the lower bound in \eqref{eq:SC} is applicable to entangled Bell-diagonal
states, while the other five bounds in \eqref{eq:SC}, \eqref{eq:VC}, and \eqref{eq:SV}
are applicable to all Bell-diagonal states. The inequality $C^2\leq V$ was also derived in Ref.~\cite{MilnJJWR14}.
As an implication of the above inequalities, any Bell-diagonal state with concurrence
larger than $(3-\sqrt{2})/(2\sqrt{2})$ is steerable by two projective measurements.
The normalized volume of the steering
ellipsoid of any separable Bell-diagonal state is bounded from  above by $1/27$, in agreement
with the result in Ref.~\cite{JevtPJR14},  while that
 of any unsteerable  Bell-diagonal
state is bounded from above by $1/(2\sqrt{2})$.

Two types of Bell-diagonal
states deserve special attention as they saturate certain inequalities in \eqref{eq:SC}, \eqref{eq:VC}, and \eqref{eq:SV}. A Werner state has the form
\begin{equation}\label{eq:Wern89}
W_f=f|\beta_{11}\rangle\langle\beta_{11} |
+\frac{1-f}{3}(I-|\beta_{11}\rangle\langle\beta_{11} |),
\end{equation}
where $0\leq f\leq 1$. Note that $f$ is equal to the singlet fraction when $f\geq1/4$. Geometrically, the Werner state  lies on a diagonal of the cube in Fig.~\ref{fig:SBDtwo};
conversely, any Bell-diagonal state lying on a diagonal of the cube is equivalent
to a Werner state under a local unitary transformation.
The correlation matrix for the Werner state has the form $T=\diag(t_1,t_2,t_3)$ with $t_1=t_2=t_3=(1-4f)/3$. Therefore,
the steering ellipsoid reduces to a sphere with radius $t_1=t_2=t_3=|4f-1|/3$; see the right plot in Fig.~\ref{fig:ellip}. In addition,
\begin{equation}
C=\max\{0,2f-1\},\quad S=\frac{2\sqrt{2}}{3}|4f-1|,\quad V=\frac{|4f-1|^3}{27}.
\end{equation}
The Werner state is steerable by two projective measurements iff $(3\sqrt{2}+2)/8<f\leq
1$.   It saturates the lower bound in \eqref{eq:SV} and, when  $f\geq\frac{1}{2}$, also the lower bound in \eqref{eq:SC} and the upper bound in \eqref{eq:VC}.

Those states
lying on an edge of the tetrahedron in Fig.~\ref{fig:SBDtwo} are called \emph{edge
states} (or rank-2
Bell-diagonal states). If an edge state has two nonzero  eigenvalues $p$
and $1-p$ with $p\geq1/2$, then $t_{11}^2=1$ and  $t_{22}^2=t_{33}^2=(2p-1)^2$ ( assuming
$t_1\geq t_2\geq |t_3|$). Therefore,
the steering ellipsoid is rotationally symmetric with the largest semiaxis equal to 1 and the other two semiaxes equal to $2p-1$; see the middle plot in Fig.~\ref{fig:ellip}. In addition, \begin{equation}
C=2p-1,\quad S=2\sqrt{1+(2p-1)^2},\quad V=(2p-1)^2.
\end{equation}
The edge state is steerable by two projective measurements whenever $p\neq 1/2$, that
is, when  it is entangled. So  entanglement is sufficient to guarantee  steering  and Bell nonlocality in this special case, which complements  Gisin's theorem \cite{Gisi91, ChenSXW15}.
In addition, the edge state saturates the  upper bounds in \eqref{eq:SC} and \eqref{eq:SV} as well as the lower bound  in \eqref{eq:VC}.

Fig.~\ref{fig:SVC2} illustrates the relations between $C, S, V$. When the concurrence $C$ is large, the three measures are closely correlated
to each other, while they tend to be more independent in the opposite scenario.
Quite surprisingly,  the normalized volume $V$ of the steering ellipsoid seems to have a closer relation with concurrence $C$ rather than the steering measure $S$. In addition, for given
concurrence $C>0$, the volume $V$  attains the maximum  when the steering measure $S$ attains the minimum, and vice
versa.

\subsection*{Steering by three projective measurements}

In this section we explore the steerability of  Bell-diagonal states under three projective measurements by the steering party. To this end, we need a criterion for  determining the compatibility of three unbiased noisy projective measurements. Fortunately, this problem has been solved  in Refs.~\cite{PalG11, YuO13}, according to which, three noisy binary observables  $\{B_{\pm|x}\}_{x=1}^3=\{(I\pm\vec{r}_{ x}\cdot\vec{\sigma})/2\}_{x=1}^3$ are compatible  iff
\begin{equation}
\label{eq:Compatibility3}
\sum_{x=0}^3|\vec{\Lambda}_x-\vec{\Lambda}_{\mathrm{FT}}|\leq4.
\end{equation}
Here $\vec{\Lambda}_0=\vec{r}_1+\vec{r}_2+\vec{r}_3$, $\vec{\Lambda}_x=2\vec{r}_x-\bm{\Lambda}_0$ for $x=1,2,3$,
and  $\vec{\Lambda}_{\mathrm{FT}}$ denotes the Fermat-Toricelli (FT) vector of $\{\vec{\Lambda}_x\}_{x=0}^3$, which is  the vector $\vec{\Lambda}$ that  minimizes the total distance  $\sum_{x=0}^3|\vec{\Lambda_x}-\vec{\Lambda}|$. In general,  $\bm{\Lambda}_{\mathrm{FT}}$ has no analytical expression~\cite{PalG11,YuO13}.

Given a Bell-diagonal state with correlation matrix $T$,
suppose Alice  performs three projective measurements $\{A_{\pm|x}\}_{x=1}^3=\{(I\pm\vec{e}_x\cdot\vec{\sigma})/2\}_{x=1}^3$. Then  Bob's steering equivalent observables are given by  $\{B_{\pm|x}\}_{x=1}^3= \{(I\pm\vec{r}_x\cdot \vec{\sigma})/2\}_{x=1}^3$, where  $\vec{r}_x=T^\rmT\vec{e}_x$ for $x=1,2,3$.
Define
\begin{equation}
\label{eq:S3}
S_3=\frac{1}{2}\max_{\vec{r}_1,\vec{r}_2,\vec{r}_3\in \mathcal{E}}\sum_{x=0}^3|\vec{\Lambda}_x-\vec{\Lambda}_{\mathrm{FT}}|
\end{equation}
 as a steering measure of  the Bell-diagonal state under three projective measurements,  where $\mathcal{E}$ is the steering ellipsoid. Then the Bell-diagonal state is steerable by three projective measurements iff $S_3>2$. In general, it is not easy to compute $S_3$. Here we shall derive a nontrivial lower bound, which is very useful for understanding the steerability of Bell-diagonal states by three projective measurements.

When  $\vec{r}_3\bot\vec{r}_{1,2}$,  the FT vector can be determined explicitly~\cite{YuO13} (note that there is a typo in Ref.~\cite{YuO13} about the sign),
\begin{equation}
\bm{\Lambda}_{\mathrm{FT}}=\frac{|\vec{r}_1-\vec{r}_2|-|\vec{r}_1+\vec{r}_2|}
{|\vec{r}_1-\vec{r}_2|+|\vec{r}_1+\vec{r}_2|}\vec{r}_3,
\end{equation}
which imply that
\begin{equation}
\sum_{x=0}^3|\vec{\Lambda}_x-\vec{\Lambda}_{\mathrm{FT}}|
=2\sqrt{(|\vec{r}_1-\vec{r}_2|+|\vec{r}_1+\vec{r}_2|)^2+4\vec{r}_3^2}.
\end{equation}

\begin{figure}[tbh]
\centering
\includegraphics[scale=0.1]{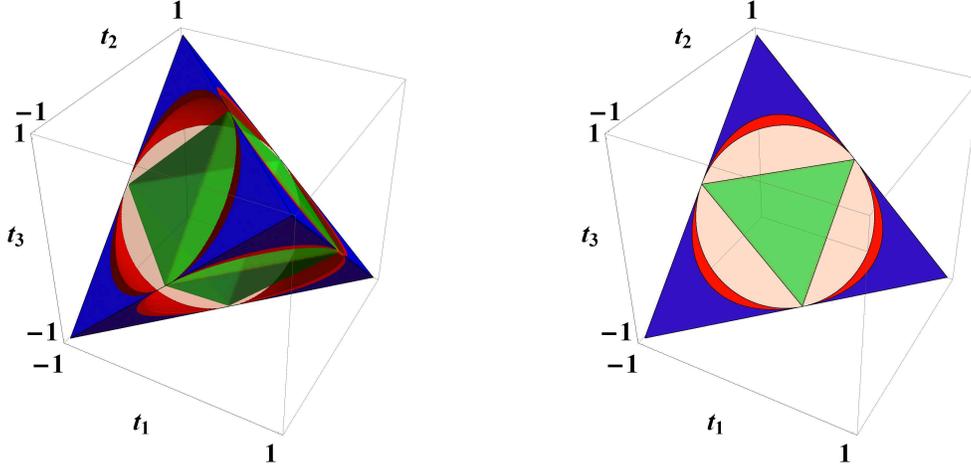}
\caption{\label{fig:SBD3} Illustration  of Bell-diagonal states steerable by three projective measurements (cf.~Fig.~\ref{fig:SBDtwo}).
(left) The regular tetrahedron
 represents the set of   Bell-diagonal states.  The octahedron
in green represents the set of separable  states. The blue regions represent
those states that  are  steerable  by two projective measurements, and the red regions represent those states that are not steerable by two projective measurements but steerable by three projective measurements as specified in the proof of Theorem~\ref{thm:SBD3}. (right)
A face of the regular tetrahedron  which represents
the set of rank-3 Bell-diagonal states.}
\end{figure}

\begin{theorem}\label{thm:SBD3} Any Bell-diagonal  state with  $\|T\|_{\mrm{F}}>
1$ is steerable by three projective measurements, where $\|T\|_{\mrm{F}}=\sqrt{\tr(T T^\rmT)}=\sqrt{\tr(T^\rmT T)}$ is the Frobenius  norm of the correlation matrix $T$.
\end{theorem}

\begin{proof}
Let $\lambda_1, \lambda_2, \lambda_3$ be the eigenvalues  of $TT^\rmT$ in nonincreasing order and $\vec{e}_1, \vec{e}_2, \vec{e}_3$  the associated orthonormal eigenvectors.   Let  $\vec{r}_x=T^\rmT\vec{e}_x$ for $x=1,2,3$. Then $\vec{r}_1,\vec{r}_2, \vec{r}_3$ are mutually orthogonal and
\begin{equation}
S_3\geq \sqrt{(|\vec{r}_1-\vec{r}_2|+|\vec{r}_1+\vec{r}_2|)^2+4\vec{r}_3^2}=2\sqrt{\lambda_1+\lambda_2+\lambda_3}=2\|T\|_{\mrm{F}}.
\end{equation}
If the Bell-diagonal state is not steerable by three projective measurements, then $S_3\leq 2$, so $\|T\|_{\mrm{F}}\leq1$.
\end{proof}

The Frobenius norm $\|T\|_{\mrm{F}}$ happens to be  the Euclidean  norm of the vector $(t_1,t_2,t_3)$ that represents the Bell-diagonal state in Figs.~\ref{fig:SBDtwo} and \ref{fig:SBD3}; its square is equal to the sum of squares of the three semiaxes of the steering ellipsoid (cf.~Ref.~\cite{CostA15}). The set of Bell-diagonal states with the same norm $\|T\|_{\mrm{F}}$ lies on a  sphere. It is clear from the above discussion that $S_3\geq 2\|T\|_{\mrm{F}}\geq  S$, so any Bell-diagonal state that is steerable by two
projective measurements is also steerable by three
projective measurements, as expected. The converse is not true in general, as illustrated in Fig.~\ref{fig:SBD3}. Consider the Werner state in \eqref{eq:Wern89} for example, we have $\|T\|_{\mrm{F}}=|4f-1|/\sqrt{3}$, so the Werner state is steerable
by three projective measurements if $1\geq f>(\sqrt{3}+1)/4$. By contrast, it is steerable by two projective measurements only if $1\geq f>(3\sqrt{2}+2)/8$.

\begin{figure*}[bth]
\centering
\includegraphics[scale=0.15]{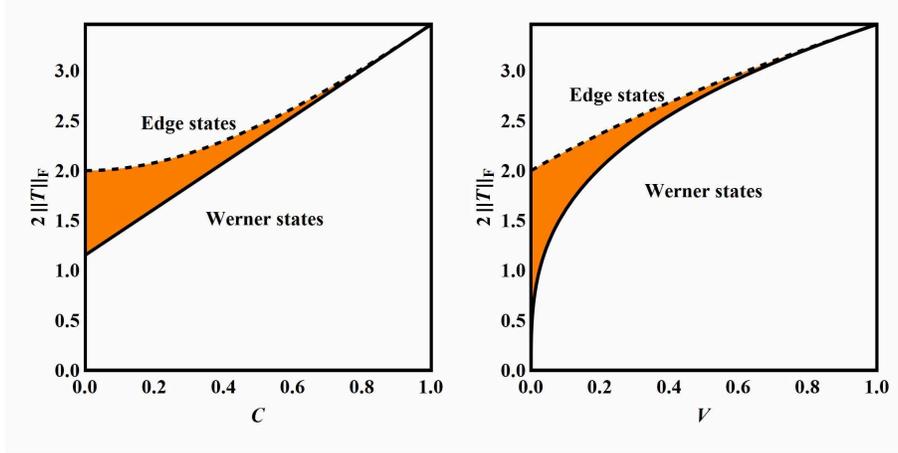}
\caption{\label{fig:SVC3}Ranges of values of $2\|T\|_{\mrm{F}}$ versus the concurrence $C$ (left) and the normalized volume $V$ of the steering ellipsoid  (right) for Bell-diagonal states. Here $\|T\|_{\mrm{F}}$ is the Frobenius norm of the correlation matrix $T$, and  $2\|T\|_{\mrm{F}}$ is a lower bound for the steering measure $S_3$, which determines the steerability of Bell-diagonal states under three projective measurements. The dashed lines represent  edge states and the solid lines represent  Werner states. }
\end{figure*}

The relations between $\|T\|_{\mrm{F}}$ and $C,V$ can be derived with similar methods
used in deriving \eqref{eq:SC} and \eqref{eq:SV},
with the results\begin{align}
\frac{1}{\sqrt{3}}(1+2C&)\leq \|T\|_{\mrm{F}}\leq\sqrt{1+2C^2},\label{eq:SC3}\\
\sqrt{3}V^{1/3}&\leq \|T\|_{\mrm{F}}\leq\sqrt{1+2V}.\label{eq:SV3}
\end{align}
Here the lower bound in \eqref{eq:SC3} is applicable to entangled
Bell-diagonal states, while the other  three bounds are applicable to all Bell-diagonal states. As in \eqref{eq:SC} and \eqref{eq:SV}, the two lower bounds are saturated   by Werner states, while the two upper bounds are saturated by edge states;
see  Fig.~\ref{fig:SVC3}. These inequalities are quite instructive to understanding the steering of Bell-diagonal states by three projective measurements given that $S_3\geq2\|T\|_{\mrm{F}}$. As an implication, any unsteerable Bell-diagonal state satisfies   $C\leq (\sqrt{3}-1)/2$ and $V\leq1/(3\sqrt{3})$.

\section*{Discussion}

In summary, we studied systematically the steerability of Bell-diagonal states by projective measurements on the steering party. In the simplest nontrivial scenario of two projective measurements, we solved the problem completely by  deriving a necessary and sufficient criterion, which has a simple geometrical interpretation. We also introduced a steering measure  and proved that it is equal to the maximal violation
of the CHSH  inequality. This conclusion implies that  a Bell-diagonal state is steerable by two projective measurements
iff it violates the CHSH inequality. In the special case of edge states, our study  shows that entanglement is sufficient to guarantee  steering and Bell nonlocality. In addition, we clarified
the relations between entanglement and steering by
deriving tight inequalities  satisfied by the concurrence, our steering measure, and the volume of the steering ellipsoid.
Finally, we explored the steerability of Bell-diagonal states under three
projective measurements. A simple  sufficient criterion was derived, which  can detect the
steerability of many states that are not steerable by two projective measurements.

Our study provided a number of instructive analytical results on steering, which are quite rare in the literature.   These results not only furnish a simple geometric picture about steering
of Bell-diagonal states, but  also offer valuable insight on the relations between entanglement, steering,
and Bell nonlocality. They may  serve as a starting point for exploring more complicated steering scenarios. In addition, our work  prompts several interesting questions, which deserve further studies. For example, is the steering criterion in Theorem~\ref{thm:SBD3}  both necessary and sufficient? Is there any upper bound on the number of measurements that are sufficient to  induce  steering for all steerable Bell-diagonal states?  We hope that these questions will stimulate further progress on the study of  steering.

\section*{Methods}

\subsection*{Concurrence and steering measure}

Here we derive the inequalities in \eqref{eq:SC}, \eqref{eq:VC}, and \eqref{eq:SV} in the main text, which characterize the relations between the concurrence~$C$, the steering measure $S$ (under two projective measurements), and the volume $V$ of the steering ellipsoid. We also determine those Bell-diagonal states that saturate these inequalities. Similar approach can  be applied to derive \eqref{eq:SC3} and \eqref{eq:SV3}, which are pertinent to steering of Bell-diagonal states by three projective measurements.

Without loss of generality, we may assume that $\rho$ has the form in~\eqref{eq:BellDiagDiag}
with $|t_3|\leq t_2\leq t_1\leq1$. Then the spectrum of $\rho$ is
given by
\begin{equation}
\frac{1}{4}\bigl\{1-t_1-t_2-t_3, 1-t_1+t_2+t_3, 1+t_1-t_2+t_3, 1+t_1+t_2-t_3
\bigr\},
\end{equation}
where the eigenvalues are arranged in nondecreasing order. The minimal and
the maximal eigenvalues are respectively given by  $p_\mrm{min}=(1-t_1-t_2-t_3)/4\geq0$  and  $p_\mrm{max}=(1+t_1+t_2-t_3)/4$.

If the Bell-diagonal state is separable, that is $C=0$, then $0\leq p_\mrm{min}\leq p_\mrm{max}\leq 1/2$~\cite{HoroH96},
which implies that
\begin{equation}
t_1+t_2+|t_3|\leq 1, \quad t_1^2+t_2^2\leq1 ,\quad |t_1t_2t_3|\leq 1/27.
\end{equation}
So the inequalities $S\leq 2\sqrt{1+C^2}$, $V\leq (1+2C)^3/27$, and  $S\leq 2\sqrt{1+V}$ in  \eqref{eq:SC}, \eqref{eq:VC}, and \eqref{eq:SV}
hold for  separable Bell-diagonal states. The inequality $S\leq 2\sqrt{1+C^2}$ is saturated iff $t_1=1$, $t_2=t_3=0$, in which case $\rho$ is an edge state with two nonzero eigenvalues equal to $1/2$. The inequality $S\leq
2\sqrt{1+V}$ is saturated under the same condition. The inequality $V\leq (1+2C)^3/27$ is saturated iff $t_1=t_2=|t_3|=1/3$, in which case $\rho$ is a Werner state which either has singlet fraction $1/2$ or is proportional to a projector of rank 3.
Here states that are equivalent to $W_f$ in \eqref{eq:Wern89} under local unitary transformations are also called Werner states. The inequality $C^2\leq V$ in  \eqref{eq:VC} is trivial for separable states; it is saturated iff $V=C=0$, that is, $t_3=0$, in which case the Bell-diagonal state lies on a coordinate plane in Fig.~\ref{fig:SBDtwo}. The inequality $2\sqrt{2}\sqrt[3]{V}\leq S$ follows from the definitions of $S$ and $V$ and is applicable to both separable and entangled states.
It is saturated iff $t_1=t_2=|t_3|$, in which case $\rho$ is
a Werner state.

If the Bell-diagonal state is entangled, then $p_\mrm{max}>1/2$,
 $C=2p_\mrm{max}-1=(t_1+t_2-t_3-1)/2$,
and $t_3=t_1+t_2-1-2C$. The positivity of $\rho$ and the requirement $|t_3|\leq
t_2\leq t_1$ lead to the following set of inequalities,
\begin{equation}
t_2\leq t_1\quad t_1+t_2\leq 1+C,\quad
t_1+2t_2\geq 1+2C.
\end{equation}
These inequalities  determine a triangular region in the parameter space of $t_1, t_2$ with the  following three vertices:
\begin{equation}
(1,C), \quad \frac{1}{2}(1+C,1+C),\quad \frac{1}{3}(1+2C,1+2C).
 \end{equation}
The maximum $1+C^2$ of $t_1^2+t_2^2$ under these constraints is attained
iff  $t_1=1,t_2=-t_3=C$, in which case the state has two nonzero eigenvalues
equal to $(1\pm C)/2$  and  is thus an edge state.  The minimum $2(1+2C)^2/9$
is attained iff $t_1=t_2=-t_3=(1+2C)/3$, in which case the state has one
eigenvalue equal
to $(1+C)/2$ and three eigenvalues  equal to $(1-C)/6$,  and  is thus a Werner
state.
By contrast, the maximum  $(1+2C)^3/27$ of $|t_1t_2t_3|$ is attained exactly when $t_1^2+t_2^2$ attains the minimum, and the minimum $C^2$ of $|t_1t_2t_3|$ is attained
when $t_1^2+t_2^2$ attains the maximum. Therefore,  \eqref{eq:SC} and \eqref{eq:VC}
hold for entangled Bell-diagonal states. As an immediate corollary, \eqref{eq:SV} also holds in this case.

In summary, the lower bound in \eqref{eq:SC} is applicable to entangled Bell-diagonal states, while the other five bounds in \eqref{eq:SC}, \eqref{eq:VC}, \eqref{eq:SV} are applicable to all Bell-diagonal states.
The two inequalities $S\leq 2\sqrt{1+C^2}$ and  $S\leq
2\sqrt{1+V}$
are saturated only for edge states. The inequality  $C^2\leq V$ is  saturated only for edge states and those states with $V=0$.
The two inequalities $2\sqrt{2}(1+2C)/3\leq S$ and $V\leq (1+2C)^3/27$ are saturated only for Werner states that have singlet fractions  at least $1/2$ or Werner states that are proportional to  rank-3 projectors. The inequality $2\sqrt{2}\sqrt[3]{V}\leq S$ is saturated only for Werner states.
In particular, among entangled Bell-diagonal states, only edge states and Werner states with singlet fractions larger than $1/2$ can saturate these inequalities.

\section*{Acknowledgements}

H.Z. is grateful to Christopher Fuchs and Johan Aberg for suggestions on the title and to Antony Milne for comments. We  gratefully thank for  the  supports  by    NSFC (Grant Nos.
 11375141, 11425522, 11434013, 91536108, 11275131). H.Z.   acknowledges financial support
from the Excellence Initiative of the German Federal and State Governments
(ZUK81) and the DFG as well as Perimeter Institute for
Theoretical Physics. Research at Perimeter Institute is
supported by the Government of Canada through Industry Canada and by the Province of Ontario through the
Ministry of Research and Innovation.

\section*{Author contributions}

H.Z.  and  Q.Q. initiated the research project and established the main results, including Theorems~1, 2, and Corollary~1. H.F., S.M.F, S.Y.L. and W.L.Y. joined some discussions and provided suggestions. Q.Q. and H.Z. wrote the manuscript with advice from H.F., S.M.F, S.Y.L. and W.L.Y.

\section*{Additional information}

\textbf{Competing financial interests:} The authors declare no competing financial interests.

\end{document}